\documentclass[reqno,10pt]{amsart}
\usepackage{graphicx}
\usepackage{amsmath}
\usepackage{amssymb} %
\usepackage{color}
\usepackage{xcolor}
\usepackage{ifpdf}
\usepackage{txfonts}
\usepackage{fleqn}
\usepackage[T1]{fontenc}
\usepackage[all]{xy}
\usepackage[colorlinks]{hyperref}
\usepackage{yfonts,eufrak}
 \usepackage{caption}
\usepackage{fancyhdr}
\usepackage [pagewise] {lineno}
\AtBeginDocument{\thanks{}}

\begin{document}

\title[Skew-constacyclic codes over $\frac{\mathbb{F}_q[v]}{\langle\,v^q-v\,\rangle}$]{ Skew-constacyclic
codes over $\frac{\mathbb{F}_q[v]}{\langle\,v^q-v\,\rangle}$}

\author[Joël  Kabore, A. Fotue-Tabue, Kenza Guenda,  Mohammed  E. Charkani]
{Joël  Kabore, Alexandre Fotue-Tabue, Kenza Guenda,  Mohammed  E.
Charkani}

\address{Joël  Kabore\newline
Department of Mathematics, University of Ouaga I Pr Joseph
Ki-Zerbo, Ouagadougou, Burkina-Faso} \email{jokabore@yahoo.fr}

\address{Alexandre Fotue-Tabue\newline
Department of mathematics, Faculty of Science,  University of
Yaounde 1, Cameroon} \email{alexfotue@gmail.com}

\address{Kenza Guenda\newline
Faculty of Engineering, University of Victoria, Canada}
\email{kguenda@uvic.ca}

\address{Mohammed  E. Charkani\newline
Department of Mathematics, Faculty of Science Dhar-Mahraz,
University of Sidi Mohamed Ben Abdellah, Fez-Atlas,30003, Morocco}
\email{mcharkani@gmail.com }

\maketitle \numberwithin{equation}{section}
\newtheorem{theo}{\bf Theorem}
\newtheorem{pro}{Proposition}[section]
\newtheorem{cor}{Corollary}
\newtheorem{lem}{Lemma}
\newtheorem{exa}{Example}
\newtheorem{rem}{Remark}[section]
\newtheorem{definition}{Definition}
\newtheorem{Notation}{Notation}
\newtheorem{Answer}{Answer}
\newtheorem{Question}{Question}
\newtheorem{Problem}{Problem}
\newtheorem{Algorithm}{{Algorithm}}
\newenvironment{Proof}[1][Proof]{\noindent {\bf Proof.\;}}{\qed}

\begin{abstract} In this paper, the investigation on the algebraic
structure of the ring
$\frac{\mathbb{F}_q[v]}{\langle\,v^q-v\,\rangle}$ and the
description of its automorphism group, enable to study the
algebraic structure of codes and their dual over this ring. We
explore the algebraic structure of skew-constacyclic codes,  by
using a linear Gray map and we determine their generator
polynomials. Necessary and sufficient conditions for the existence
of self-dual skew cyclic and self-dual skew negacyclic codes over
$\frac{\mathbb{F}_q[v]}{\langle\,v^q-v\,\rangle}$ are given.
 \vspace{0.15 cm}


\vspace{0.15cm}

\end{abstract}

\section{Introduction}

Codes over finite commutative chain rings have been extensively
studied. In the last years, some special rings which are non-chain
rings are used as alphabet for codes. Recently, using a
(non-commutative) skew polynomial ring, which is a particular case
of well-known Öre polynomial rings, D. Boucher et al. developed
some classes of linear codes called skew-constacyclic codes
\cite{BGU, BSU,BU, BU2}. In analogous to classical constacyclic
codes, skew constacyclic codes over a finite commutative ring have
rich algebraic structures and can be seen as left-submodules  of a
certain class of module. These codes have received much attention
in recent years \cite{BL,JLU, SAS} and have been studied for
various automorphisms over some non-chain rings \cite{GMF, GSY,
MGGSS, ZW}.

 In \cite{AAS,AS}, Abualrub et al. studied cyclic, and skew cyclic codes over
$\mathbb{F}_2[v]/\langle\,v^2-v\,\rangle.$ Bayram et al.
investigated cyclic codes over
$\mathbb{Z}_3[v]/\langle\,v^3-v\,\rangle$ \cite{BS1} and
constacyclic codes over $\mathbb{F}_p[v]/\langle\,v^p-v\,\rangle$
\cite{BS}. In \cite{DCE}, Dertli et al. explored skew-constacyclic
and skew-quasi constacyclic codes over
$\mathbb{Z}_3[v]/\langle\,v^3-v\,\rangle.$ Recently, Shi et al.,
in \cite{SYS}, studied skew cyclic codes over
$\mathbb{F}_{p^s}[v]/\langle\,v^m-v\,\rangle$, where $p$ is a
prime and $m-1$ divides $p-1.$ The motivation for studying this
ring is lying under the facts that first this ring is as a natural
generalization of codes over the ring
$\mathbb{F}_{p}[v]/\langle\,v^p-v\,\rangle.$  Second important
fact regarding this ring is that a linear Gray map is defined and
hence codes over fields are obtained. To the best of our
knowledge, the study of skew constacyclic codes over
$\mathbb{F}_{q}[v]/\langle\,v^q-v\,\rangle$ has not been
considered by any coding scientist. Our objection is to
characterize self-dual skew-constacyclic codes over
$\mathbb{F}_{q}[v]/\langle\,v^q-v\,\rangle.$

This paper is organized as follows. In Section \ref{sec:2}, we
first give some properties about the ring
$\mathbb{F}_{q}[v]/\langle\,v^q-v\,\rangle$ and describe its
ring-automorphism group. In Section \ref{sec:3}, we explore linear
codes over $\mathbb{F}_{q}[v]/\langle\,v^q-v\,\rangle$ using a
linear Gray map and we show that the Gray-image of any skew
constacyclic code of length $n$ over
$\mathbb{F}_{q}[v]/\langle\,v^q-v\,\rangle$ under this linear Gray
map is a skew multi-twisted code of length $qn$ over
$\mathbb{F}_q.$ In Section \ref{sec:4}, we determine the structure
of skew-constacyclic codes and characterize self-dual
skew-constacyclic codes over
$\mathbb{F}_{q}[v]/\langle\,v^q-v\,\rangle.$ We also give
necessary and sufficient conditions for the existence of self-dual
skew cyclic and self-dual skew negacyclic codes over
$\mathbb{F}_{q}[v]/\langle\,v^q-v\,\rangle.$

\section{Ring-automorphism group of
$\mathbb{F}_{q}[v]/ \langle v^{q}-v \rangle$}\label{sec:2}

Throughout this paper, $\mathbb{F}_q$ is a finite field with $q$
elements and $\beta$ is a generator of multiplicative group
$\mathbb{F}_q\backslash\{0\}.$ Write $\mathbb{F}_q :=\{\alpha_0,
\alpha_1,\cdots,\alpha_{q-1}\}$ where $\alpha_0=0,
\alpha_i=\beta^{i}$ for all $0 \leq i \leq q-1.$ We denote the
ring $\mathbb{F}_{q}[v]/\langle\,v^q-v\,\rangle$ by $R_q$ and
$\mathcal{U}(R_q)$ its unit group. Obviously, $\mathbb{F}_q$ is a
subring of the ring $R_q,$ and $R_q$ is a vector space over
$\mathbb{F}_q$ with basis
$\{1,\overline{v},\cdots,\overline{v}^{q-1}\}$ where
$\overline{v}:=v+\langle\,v^q-v\,\rangle.$ Thus
$R_q=\mathbb{F}_q\oplus\mathbb{F}_q\overline{v}\oplus\cdots\oplus\mathbb{F}_q\overline{v}^{q-1}$
and any element $r$ in $R_q$ can be uniquely written as:
$r=r_0+r_1\overline{v}+\cdots+r_q \overline{v}^{q-1}$ with $r_i
\in \mathbb{F}_q$, for all $0 \leq i \leq q-1.$ Moreover, the ring
$R_q$ is a non-chain principal ideal ring with maximal ideals
$\langle\,v-\alpha_i\,\rangle/\langle\,v^q-v\,\rangle$ ($0\leq i
<q$), since the factorization of $v^q-v$ in $\mathbb{F}_q[v]$ is
$v^q-v=(v-\alpha_0)(v-\alpha_1)\cdots(v-\alpha_{q-1}).$

An element $\eta$ of $R_q$ is called \emph{idempotent} if
$\eta^2=\eta;$ two idempotents $\eta_1, \eta_2$ are said to be
\emph{orthogonal} if $\eta_1\eta_2=0.$ An idempotent of $R_q$ is
said to be \emph{primitive} if it is non-zero and it can not be
written as sum of non-zero orthogonal idempotents. A set
$\{\eta_0,\eta_1,\cdots ,\eta_{q-1}\}$ of idempotents of $R_q$ is
\emph{complete} if $\sum_{i=0}^{q-1} \eta_i=1.$ Let $f_i(v):= v -
\alpha_i,$ and $\widehat{f_i}(v) := \frac{v^q-v}{f_i(v)},$ where
$i = 0, 1,\cdots, q - 1.$ Then there exist $a_i(v)$ and  $b_i(v)$
in $\mathbb{F}_q[v]$ such that $a_i(v)f_i(v) +
b_i(v)\widehat{f_i}(v) = 1.$ Let $\eta_i =
b_i(\overline{v})\widehat{f_i}(\overline{v}),$ then $\eta_i^2 =
\eta_i,$ $\eta_i\eta_j = 0$ and $\sum\limits_{i=0}^{q-1}\eta_i =
1,$ where $0\leq i\neq j \leq q-1.$ It is easy to see that any
complete set of idempotents in $R_q$ is a basis of
$\mathbb{F}_q$-vector space. Therefore, any element $r$ of $R_q$
can be represented as: $r= r_0 \eta_0+ r_1 \eta_1+\cdots
+r_{q-1}\eta_{q-1}$ \cite{SYS}.

For $i=0,1,\cdots,q-1,$ we consider the map
\begin{align}\label{idi}
\begin{array}{cccc}
  \phi_i: & R_q & \rightarrow & \left(\mathbb{F}_q;+,\cdot; 1,0\right) \\
    & \sum\limits_{i=0}^{q-1}a_i\eta_{i} & \mapsto & a_i,
\end{array}
\end{align}
which is a ring-epimorphism. We denote by $\ast$ the componentwise
multiplication (or Schur product) and by $+$ the componentwise
addition on $(\mathbb{F}_q)^q,$ i.e. for $\textbf{x} :=
(x_0,x_1,\cdots,x_{q-1}), \textbf{y} :=
(y_0,y_1,\cdots,y_{q-1})\in (\mathbb{F}_q)^q,$ we put $$\textbf{x}
\ast \textbf{y} := (x_0y_0,x_1y_1,\cdots,x_{q-1}y_{q-1}) \in
(\mathbb{F}_q)^q ,$$
$$\textbf{x} + \textbf{y} :=
(x_0+y_0,x_1+y_1,\cdots,x_{q-1}+y_{q-1}) \in (\mathbb{F}_q)^q .$$
By the remainder Chinese theorem, the following map
\begin{align}\label{id}
\begin{array}{cccc}
  \phi: & R_q & \rightarrow & \left((\mathbb{F}_q)^q;+,\ast; \textbf{1},\textbf{0}\right) \\
    & a & \mapsto &
    \left(\phi_0(a),\phi_1(a),\cdots,\phi_{q-1}(a)\right),
\end{array}
\end{align} is a ring-isomorphism, where $\textbf{1}:=(1,1,\cdots,1)$ and
$\textbf{0}:=(0,0,\cdots,0).$ It is easy to see that the ring
$R_q$ is a principal ideal ring whose $2^q$ ideals are
$\mathrm{I}_{A}:=\langle\sum\limits_{i\in A}\eta_i\rangle$ where
$A$ is a subset of $\{0;1;\cdots;q-1\}.$ For any $a\in R_q,$ we
set $\texttt{Supp}(a):=\{i \in \{0;1;\cdots;q-1\} : \phi_i(a) \neq
0\};$ then $\langle\,a\,\rangle=I_{\texttt{Supp}(a)}$ and
$|\langle a\rangle|=q^{|\texttt{Supp}(a)|}.$ Moreover, the group
of units of $R_q$ is described as follows:
$$\mathcal{U}(R_q)=\left\{\sum\limits_{i=0}^{q-1}a_i\eta_{i}\;:\;a_i\neq
0\text{ for all }0\leq i \leq q-1\right\},$$ and so
$|\mathcal{U}(R_q)|=(q-1)^q.$ For instance, the ring $R_3$ has $8$
units given by: $1, 2, 1+v^2, 1+v+2v^2, 1+2v+2v^2, 2+v+v^2,
2+2v+v^2, 2+2 v^2$ \cite{BS1}.

\begin{lem}\label{idem} The ring $R_q$ admits
a unique complete set $\{\eta_0,\eta_1,\cdots ,\eta_{q-1}\}$ of
primitive pairwise orthogonal idempotents.
\end{lem}

\begin{Proof} The only complete set of primitive
pairwise orthogonal idempotents in the ring
$\left((\mathbb{F}_q)^q;+,\ast; \textbf{1},\textbf{0}\right)$ is
$\{e_0,e_1,\cdots ,e_{q-1}\}$ where $e_0:=(1,0,\cdots,0); $
$e_1:=(0,1,0,\cdots,0);$
$e_2:=(0,0,1,0,\cdots,0);\cdots;e_{q-1}:=(0,0,\cdots,0,1).$ From
the ring-isomorphism (\ref{id}), the set $\{\eta_0,\eta_1,\cdots
,\eta_{q-1}\}$ is also the complete set of primitive pairwise
orthogonal idempotents in the ring $R_q,$ where
$\phi(\eta_i)=e_i,$ for all $0\leq i\leq q-1.$
\end{Proof}

\begin{theo}
Let $\theta$ be a ring-automorphism of $\mathbb{F}_q$ and $\sigma$
is a permutation of $\{0,1,\cdots,q-1\}.$ Then the map
\begin{align}
\begin{array}{cccc}
\Theta_{\theta, \sigma}: & R_q& \rightarrow & R_q\\
&\sum\limits_{i=0}^{q-1}a_i\eta_{i} &\mapsto &
\sum\limits_{i=0}^{q-1}\theta(a_i)\eta_{\sigma(i)}.
\end{array}
\end{align}
is a ring-automorphism group of $R_q.$ Moreover
$$\texttt{Aut}(R_q)=\left\{\Theta_{\theta, \sigma}\,:\, \theta \in \texttt{Aut}(\mathbb{F}_q),
\sigma\in\mathbb{S}_{q}\right\},$$ where $\mathbb{S}_q$ is the
group of permutations of $\{0,1,\cdots,q-1\},$ and
$$|\texttt{Aut}(R_q)|=r\cdot q\cdot(q-1)\cdot(q-2)\cdot\ldots\cdot 2\cdot 1,$$ where $q=p^r$ with $p$ a prime number.
\end{theo}
\begin{Proof} First, for all $\theta\in\texttt{Aut}(\mathbb{F}_q)$ and
$\sigma\in\mathbb{S}_{q},$ it is obvious to check that
$\Theta_{\theta, \sigma}$ is a ring-automorphism of $R_q.$
Therefore $\left\{\Theta_{\theta, \sigma}\,:\, \theta \in
\texttt{Aut}(\mathbb{F}_q),
\sigma\in\mathbb{S}_{q}\right\}\subseteq\texttt{Aut}(R_q).$
Inversely, if $\Theta\in\texttt{Aut}(R_q)$ then it is clear that
the restriction of $\Theta$ over $\mathbb{F}_q$ is a
ring-automorphism $\theta$ of $\mathbb{F}_q.$ Thus for any
$a:=\sum\limits_{i=0}^{q-1}a_i\eta_{i}$ in $R_q,$ we have
$\Theta(a)=\sum\limits_{i=0}^{q-1}\theta(a_i)\Theta(\eta_{i}).$
Now the set $$\{\Theta(\eta_{i})\;:\;0\leq i\leq q-1\}$$ is
another complete set of primitive pairwise orthogonal idempotents
in $R_q.$ From Lemma\,\ref{idem}, it follows that there is a
permutation $\sigma$ of $\{0,1,\cdots,q-1\}$ such that
$\Theta(\eta_{i})=\eta_{\sigma(i)}.$ Hence
$\Theta(a)=\sum\limits_{i=0}^{q-1}\theta(a_i)\eta_{\sigma(i)}$ and
$\texttt{Aut}(R_q)=\left\{\Theta_{\theta, \sigma}\,:\, \theta \in
\texttt{Aut}(\mathbb{F}_q), \sigma\in\mathbb{S}_{q}\right\}.$
Finally, $\Theta_{\theta', \sigma'}\circ\Theta_{\theta,
\sigma}=\Theta_{\theta'\circ\theta, \sigma'\circ\sigma},$ for all
$\theta',\theta\in\texttt{Aut}(\mathbb{F}_q)$ and
$\sigma',\sigma\in\mathbb{S}_{q}.$
\end{Proof}

\begin{exa} The ring-automorphism group of $R_2$ and $R_4:$
\begin{enumerate}
\item The idempotents in $R_2$ are $\eta_0:=\overline{v}$ and
$\eta_1:=\overline{v}+1.$ Since
$\mathbb{S}_2:=\{\texttt{id},\sigma:=(01)\}$ and
$\texttt{Aut}(\mathbb{F}_2)=\{\texttt{Id}\}$, out of the identity
map, the only ring-automorphism over $R_2$ is given by:
$$\begin{array}{c c c c}
\Theta_{\texttt{Id}, \sigma}: & R_2& \longrightarrow & R_2\\
& a\eta_0+ b\eta_1&\longmapsto & a\eta_1+ b\eta_0.
\end{array}
$$
The automorphism $\Theta_{\texttt{Id}, \sigma}$ is used in
\cite{AAS} to study skew $\Theta_{\texttt{Id}, \sigma}$-cyclic
codes over $R_2.$
 \item The complete set of
primitive pairwise orthogonal idempotents of $R_4$ is given by:
$$\eta_0=\overline{v}^3+1; \eta_1= \overline{v}^3+\overline{v}+1; \eta_2=\overline{v}^3+\alpha \overline{v}^2+\alpha^2\overline{v};
\eta_4=\overline{v}^3+\alpha^2\overline{v}^2+\alpha
\overline{v}.$$  A ring-automorphism over $R_4$ is given by:
$$\begin{array}{c c c c}
\Theta_{\theta, \sigma_i}: & R_4& \longrightarrow & R_4\\
& \sum\limits_{i=0}^{3}a_i\eta_i &\longmapsto &
\sum\limits_{i=0}^{3}\theta(a_i)\eta_{\sigma(i)},
\end{array}
$$ for all $\theta\in\texttt{Aut}(\mathbb{F}_4)$ and $\sigma\in\mathbb{S}_4.$ Thus there exist exactly $48$ automorphisms over $R_4.$
\end{enumerate}
\end{exa}

\section{Linear Gray map and linear codes over $R_q$}\label{sec:3}

The ring-morphism $\phi_i: R_q \rightarrow \mathbb{F}_q $ defined
in (\ref{idi}) is naturally extended to $(R_q)^n$ as follows:
\begin{align}\begin{array}{c c c c}
\Phi_i:& (R_q)^n & \rightarrow & (\mathbb{F}_q)^{n}\\
& \left(a_0,a_1,\cdots,a_{n-1}\right) & \mapsto &
\left(\phi_i(a_0),\phi_i(a_1),\cdots,\phi_i(a_{n-1})\right),
\end{array}
\end{align}
which is an epimorphism of vector spaces over $\mathbb{F}_q.$ We
define a linear Gray map over $R_q$ to be
\begin{align}\begin{array}{c c c c}
\Phi:& (R_q)^n& \rightarrow & \left((\mathbb{F}_q)^n\right)^{q}\\
& \textbf{a} &\mapsto & \left(\Phi_0(\textbf{a}),
\Phi_1(\textbf{a}),\cdots,\Phi_{q-1}(\textbf{a})\right).
\end{array}
\end{align}
This map $\Phi:(R_q)^n \rightarrow
\left((\mathbb{F}_q)^n\right)^{q}$ is an isomorphism of vector
spaces over $\mathbb{F}_q.$

\begin{definition}
The Gray weight of any element $\textbf{a}$ in $(R_q)^n$ is
defined as: $W_G(\textbf{a})= W_{H}(\Phi(\textbf{a})),$ where
$W_H$ denotes the Hamming weight over $\mathbb{F}_q.$ The Gray
distance between two elements $\textbf{a}_1, \textbf{a}_2$ in
$(R_q)^n$ is given by: $d_{G}(\textbf{a}_1,
\textbf{a}_2)=W_G(\textbf{a}_1-\textbf{a}_2).$
\end{definition}

It is obvious that the linear Gray map $\Phi$ is a weight
preserving map from $((R_q)^n, W_G)$ to
$\biggl(\left((\mathbb{F}_q)^n\right)^{q}, W_H\biggr)$ and
$W_G(\textbf{a})=\sum\limits_{i=0}^{q-1}W_{H}(\Phi_i(\textbf{a}))$
for $\textbf{a}\in (R_q)^n.$

A \emph{code} $C$ of length $n$ over $R_q$ is a nonempty subset of
$(R_q)^n.$ If in addition the code is a submodule of $R^n_q$, it
is called \emph{linear code}. The Euclidean inner product between
two elements $\textbf{a}=(a_0, a_1,\cdots ,a_{n-1}),$ and
$\textbf{b}=(b_0,b_1,\cdots, b_{n-1})$ in $(R_q)^n$ is defined by:
$\textbf{a}\cdot \textbf{b}= \sum\limits_{i=0}^{n-1} a_i b_i.$ The
\emph{Euclidean dual code} (shortly dual code) of a code $C$ of
length $n$ over $R_q$ is defined as: $$C^{\perp
}=\left\{\textbf{a} \in (R_q)^n~ ; ~ \textbf{a}\cdot \textbf{b}=0,
\text{ for all }~ \textbf{b} \in C\right\}.$$ Recall that
Euclidean dual of linear code  of length $n$ over $R_q$ is a
linear code  of length $n$ over $R_q.$

\begin{exa} The linear codes of length $1$ over $R_q$ are
ideals $I_A$ of  $R_q,$ where $A$ is a subset of
$\{0;1;\cdots;q-1\}.$ Thus for all $a\in R_q,$ $W_G(a)=|A|$ and
$\langle
a\rangle^{\perp}=\left\langle\sum\limits_{i\in\overline{A}}\eta_i\right\rangle$
where $A:=\texttt{Supp}(a)$ and
$\overline{A}:=\left\{i\in\{0;1;\cdots;q-1\}\;:\;i\not\in
A\right\}.$ Therefore $I_A$ is a linear code of length $1$ over
$R_q$ with Gray weight $|A|$ and $\Phi(I_A)$ is a $[q,1,|A|]$-code
over $\mathbb{F}_q.$
\end{exa}

\begin{pro}\label{pro:self}
Let $C$ be a linear code of length $n$ over $R_q.$ Then
$\Phi(C^\perp)=\left(\Phi(C)\right)^\perp,$ where
$\left(\Phi(C)\right)^\perp$ denotes the ordinary dual of
$\Phi(C)$ as a linear code over $\mathbb{F}_q.$ Moreover, $C$ is a
self-dual code of length $n$ over $R_q$ if and only if $\Phi(C)$
is a self-dual code of length $qn$ over $\mathbb{F}_q.$
\end{pro}

\begin{proof} For all $\textbf{a}$ and $\textbf{b}$ in $(R_q)^n,$
we have $$ \textbf{a}\cdot\textbf{b}=0~~\Rightarrow ~~
\Phi(\textbf{a})\cdot\Phi(\textbf{b})=0, $$ where
$\Phi(\textbf{a})\cdot\Phi(\textbf{b})$ denotes the usual standard
inner product in $\left((\mathbb{F}_q)^n\right)^{q}.$ An immediate
consequence is the inclusion:
$$\Phi(C^\perp)\subseteq\left(\Phi(C)\right)^\perp.$$
Combining this with the fact that in Frobenius rings,
$|C||C^\perp|=|R_q|^n,$ we get
$\Phi(C^\perp)=\left(\Phi(C)\right)^\perp.$ Finally, since $\Phi$
is an isomorphism of vector spaces over $\mathbb{F}_q,$ the
equivalence is an immediate consequence of the equality
$\Phi(C^\perp)=\left(\Phi(C)\right)^\perp.$
\end{proof}

Since $R_q=\eta_0 R_q \oplus \eta_1 R_q \oplus \cdots \oplus
\eta_{q-1} R_q$, it follows that $$(R_q)^n=\eta_0 (R_q)^n \oplus
\eta_1 (R_q)^n \oplus \cdots \oplus \eta_{q-1} (R_q)^n.$$ Let $C$
be a linear code of length $n$ over $R_q$ and
$\textbf{a}=(a_0,a_1,\cdots,a_{n-1})\in C.$ Then $a_i=
\sum\limits_{j=0}^{q-1}\eta_j\phi_j(a_i)$ and
$\textbf{a}=\sum\limits_{i=0}^{q-1}\eta_i\Phi_i(\textbf{a})=
\sum\limits_{i=0}^{q-1}\eta_i \textbf{a}_i,$ where $\textbf{a}_i:=
\Phi_i(\textbf{a}).$ We let
\begin{align}\label{ci}C_i:=\Phi_i\left( C \right),\end{align} for $0\leq i\leq q-1.$
For this, it is straightforward to see that $C_0,
C_1,\cdots,C_{q-1}$ are linear codes of length $n$ over
$\mathbb{F}_q$ and $$C= \eta_0 C_0 \oplus \eta_1 C_1
\oplus\cdots\oplus \eta_{q-1}C_{q-1}.$$

\begin{pro}\label{pro:Cself}
Let $C= \eta_0 C_0 \oplus \eta_1 C_1 \oplus\cdots\oplus
\eta_{q-1}C_{q-1}$ be a linear code over $R_q$ of length $n.$ Then
$$C^{\perp }= \eta_0 C_0^{\perp } \oplus \eta_1 C_1^{\perp } \oplus\cdots\oplus \eta_{q-1}C_{q-1}^{\perp
}.$$  Moreover, $C$ is a self-dual code of length $n$ over $R_q$
if and only if $C_i$ is a self-dual code of length $n$ over
$\mathbb{F}_q$ for all $0\leq i \leq q-1.$
\end{pro}

\begin{Proof} Let $D=\eta_0 C_0^{\perp } \oplus \eta_1
C_1^{\perp }\oplus\cdots\oplus \eta_{q-1}C_{q-1}^{\perp }.$ It is
clear that $D \subseteq C^{\perp }.$ Since $R_q$ is a principal
ideal ring, consequently $R_q$ is a Frobenius ring, then from
\cite{W}, we have $|C^{\perp
}|=\frac{|R_q|^n}{|C|}=\frac{q^{qn}}{|C|}.$ Thus
\begin{eqnarray*}
 |D| &=& |C_0^{\perp }||C_1^{\perp }|\cdots|C_{q-1}^{\perp }| \\
    &=& \frac{|\mathbb{F}_q|^n}{|C_0|}\frac{|\mathbb{F}_q|^n}{|C_1|}\cdots\frac{|\mathbb{F}_q|^n}{|C_{q-1}|} \\
    &=& \frac{q^{qn}}{|C_0||C_1|\cdots|C_{q-1}|}\\
    &=& \frac{q^{qn}}{|C|}=|C^{\perp }|.
\end{eqnarray*}
Hence $C^{\perp }=D.$

On the other hand, if $C$ is self-dual, then for all $0\leq i\leq
q-1,$ the code $C_i$ is also self-dual, by Definition (\ref{ci})
of the code $C_i.$ Inversely, if each $C_i$ is also self-dual,
then $C$ is also self-dual, since $C= \eta_0 C_0 \oplus \eta_1 C_1
\oplus\cdots\oplus \eta_{q-1}C_{q-1}.$
\end{Proof}

Let   $\textbf{v}_1, \cdots , \textbf{v}_k$ be vectors in
$(R_q)^n.$ The vectors  $\textbf{v}_1, \cdots , \textbf{v}_k$  are
said to \emph{modular independent}, if $\Phi_i(\textbf{v}_1),
\cdots , \Phi_i(\textbf{v}_k)$ are linearly independent for some
$i.$ The vectors  $\textbf{v}_1, \cdots , \textbf{v}_k$  are said
to \emph{independent}, if $\sum \alpha_j\textbf{v}_j=\textbf{0}$
implies that $\alpha_j\textbf{v}_j=\textbf{0}$ for all $j.$

Let $C= \eta_0 C_0 \oplus \eta_1 C_1 \oplus\cdots\oplus
\eta_{q-1}C_{q-1}$ be a linear code of length $n$ over $R_q.$ Then
from obtained results in \cite{DL}, it follows that
 $\texttt{rank}_{R_q}(C)=\texttt{max}\{\texttt{dim}_{\mathbb{F}_q}(C_i)\;:\;0\leq
i \leq q-1\},$ $\Phi(C)$ is a linear code of length $qn$ over
$\mathbb{F}_q$ and $\Phi(C)= C_0 \times C_1 \times\cdots\times
C_{q-1}$ with $C_i$ be a linear code of length $n$ over
$\mathbb{F}_q.$ In \cite{DL}, the codewords $\textbf{c}_1, \cdots
, \textbf{c}_k$ form a \emph{basis} of $C,$ if they are
independent, modular independent and generate $C.$ We well-know
that the rows of a generator matrix for $C$ form a basis of $C.$
The following result is a direct consequence from the above
discussion.

\begin{theo}\label{thm:gen}
Let $C= \eta_0 C_0 \oplus \eta_1 C_1
\oplus\cdots\oplus\eta_{q-1}C_{q-1}$ be a linear code over $R_q$
of rank $k,$ and $\mathrm{G}_i$ is a generator matrix of linear
code $C_i$ over $\mathbb{F}_q$ of dimension $k_i,$ for $0\leq
i\leq q-1.$ Then a generator matrix $\mathrm{G}$ for $C$ is given
by:
$$
\mathrm{G}=\eta_0 \widetilde{\mathrm{G}}_0+ \eta_1
\widetilde{\mathrm{G}}_1+ \cdots
+\eta_{q-1}\widetilde{\mathrm{G}}_{q-1}
$$
where $\widetilde{\mathrm{G}}:=\left(%
\begin{array}{c}
  \mathrm{G}_i \\
  \mathrm{O}_{(k-k_i)\times n} \\
\end{array}%
\right)$ with $\mathrm{O}_{(k-k_i)\times n}$ is the $(k-k_i)\times
n$-matrix whose entries are zeros, and a generator matrix for
$\Phi(C)$ is given by:
$$
\Phi(\mathrm{G}):=\left(
\begin{array}{ccccc}
\mathrm{G}_0 & 0 & 0 & \cdots  & 0 \\
0 & \mathrm{G}_1 & 0 & \cdots  & 0 \\
\vdots  & \ddots  & \ddots  & \ddots  & \vdots  \\
0 & \cdots  & 0 & \mathrm{G}_{q-2} & 0 \\
0 & \cdots  & 0 & 0 & \mathrm{G}_{q-1}
\end{array}%
\right).
$$
\end{theo}

\begin{cor}\label{cor:para} Let $C$ be a linear code of length $n$ over
$R_q$ such that $$C= \eta_0 C_0 \oplus \eta_1 C_1
\oplus\cdots\oplus\eta_{q-1}C_{q-1},$$ where $C_{i}$ is an $[n,
k_i, d_i]$-linear code over $\mathbb{F}_q.$ Then $\Phi(C)$ is a
$\left[qn, \sum\limits_{i=0}^{q-1}k_i, \underset{0\leq i \leq
q-1}{\min }\,d_i\right]$-linear code over $\mathbb{F}_q.$
\end{cor}

\begin{exa}
Consider the ring $\mathbb{F}_4[v]/\langle\,v^4-v\,\rangle,$ where
$\mathbb{F}_4=\{0,1,\alpha,\alpha^2\}$ is the finite field with
four elements such that $\alpha^2+\alpha+1=0.$ The complete set of
primitive pairwise orthogonal idempotents of $R_4$ is given by:
$$\eta_0=\overline{v}^3+1; \eta_1= \overline{v}^3+\overline{v}+1; \eta_2=\overline{v}^3+\alpha \overline{v}^2+\alpha^2\overline{v};
\eta_4=\overline{v}^3+\alpha^2\overline{v}^2+\alpha \overline{v}.$$ We let \\
$$G_1:=\left(\begin{array}{llllll}
1 & 0 & 0 & \alpha^2 & \alpha^2 & 1\\
0 & 1 & 0 & \alpha^2 & 0 &\alpha \\
0 & 0 & 1 & 1 & \alpha & \alpha \\
\end{array}
\right);~~ G_2:=\left(\begin{array}{llllll}
1&0&0&\alpha&\alpha&1\\
0&1&0&\alpha&0&\alpha^2\\
0&0&1&1&\alpha^2&\alpha^2\\
\end{array}
\right);$$ and $G_3:=\left(\begin{array}{llllll}
1&0&0&1&0&0\\
0&1&0&0&1&0\\
0&0&1&0&0&1\\
\end{array}
\right). $ Let $C= \eta_0 C_0 \oplus \eta_1 C_1 \oplus\eta_2
C_2\oplus\eta_{3}C_3$ be a linear code of length 6 over $R_4$;
where $C_0, C_1,C_2$ are self dual $[6,3,3]$-codes of length 6
over $\mathbb{F}_4$  with generator matrix $G_1, G_1, G_2$
respectively. The code $C_3$ generated by $G_3$ is a self-dual
$[6,3,2]$-code over $\mathbb{F}_4$ \cite{JLX}. From
Proposition\,\ref{pro:Cself}, $C$ is a self-dual code of length
$6$ over $R_4.$ From Proposition\,\ref{pro:self},
Theorem\,\ref{thm:gen} and Corollary\,\ref{cor:para}, the code
$\Phi(C)$ is a self-dual $[24,12,2]$-code with generator matrix
$\Phi(G)=\left(
\begin{array}{l l l l}
G_1&0& 0& 0\\
0 & G_1&0&0\\
0 & 0&G_2&0\\
0&0&0&G_3
\end{array}
\right).$
\end{exa}

\section{Skew-constacyclic codes over $R_q$}\label{sec:4}

Let $\Theta$ be a ring-automorphism of $R_q$ such that
$\Theta=\Theta_{\theta,\sigma}$ where $\theta\in
\texttt{Aut}(\mathbb{F}_q)$, $\sigma\in\mathbb{S}_q.$ The
$\theta$-skew polynomial ring over $\mathbb{F}_q,$ denoted
$\mathbb{F}_q[x; \theta],$ is right Euclidean domain \cite{Jac96}.
The \emph{least common left multiple} (\texttt{lclm}) of nonzero
$f_1$ and $f_2,$ in $ \mathbb{F}_q[x; \theta],$ denoted
$\texttt{lclm}(f_1, f_2),$ is the unique monic polynomial $h\in
\mathbb{F}_q[x; \theta]$ of lowest degree such that there exist
$u_1, u_2 \in \mathbb{F}_q[x; \theta]$ with $h = u_1f_1$ and $h =
u_2f_2.$ The skew polynomial ring $R_q[x, \Theta]$ is the set
$R_q[x]$ of formal polynomials over $R_q$, where the addition is
defined as the usual addition of polynomials and the
multiplication is defined  using the rule  $x a = \Theta(a)x$,
which is extended to all elements of $R_q[x, \Theta]$ by
associativity and distributivity. An element in $R_q[x, \Theta]$
is called a \emph{skew polynomial}.

Let $\lambda$ be a unit of $R_q.$ For a given automorphism
$\Theta$ of $R_q$, we define the $\Theta$-$\lambda$- constacyclic
shift $\texttt{T}_{\Theta,\lambda}$ on $R^n_q$ by:
$$\texttt{T}_{\Theta,\lambda}((a_0,a_1,\cdots,a_{n-1}))=(\lambda \Theta(a_{n-1}), \Theta(a_0),\cdots, \Theta(a_{n-2})).$$
A linear code $C$ of length $n$ over $R_q$ is said to be
\emph{skew-$\Theta$-$\lambda$-constacyclic}, if
$\texttt{T}_{\Theta,\lambda}(C)=C.$ Similarly to classical
constacyclic codes over finite rings,
skew-$\Theta$-$\lambda$-constacyclic codes of length $n$ over
$R_q$ are identified with the left $R_q[x,\Theta]$-submodule of
$R_q[x, \Theta]/\langle\,x^n-\lambda\,\rangle$ by the
identification:
$$ \varphi: (a_0,a_1,\cdots,a_{n-1})\longmapsto a_0+a_1x+\cdots+a_{n-1}x^{n-1}.$$
Indeed, it is straightforward to see that the set $R_q[x,
\Theta]/\langle\,x^n-\lambda\,\rangle$ is a left
$R_q[x,\Theta]$-module under the multiplication defined by
$$f(x)\left(g(x)+ \langle\,x^n- \lambda\,\rangle\right)= f(x)g(x)+ \langle\,x^n- \lambda\,\rangle,$$
with $f(x), g(x) \in R_q[x, \Theta]/\langle\,x^n-\lambda\,\rangle$
and by analogous methods that have been used for skew constacyclic
codes over finite fields \cite{BU2}, we have the following fact.

\begin{lem}
Let $\lambda$ be a unit of $R_q$, $\Theta$ an automorphism of
$R_q$ and $C$ be a linear code of length $n$ over $R_q.$ Then $C$
is a skew-$\Theta$-$\lambda$-constacyclic code if and only if
$\varphi(C)$ is a left $R[x, \Theta]$-submodule of $R_q[x,
\Theta]/\langle\,x^n-\lambda\,\rangle.$
\end{lem}

In order to characterize skew-$\Theta$-$\lambda$-constacyclic
codes over $R_q$, we recall some well-known results about
skew-constacyclic codes over finite fields \cite{BU2, BU3, GMF,
SAS}. The skew reciprocal polynomial of a polynomial
$g=\sum\limits_{j=0}^{k}a_j x^j \in \mathbb{F}_q[x,\theta]$ of
degree $k$ denoted by $g^{*}$ is defined as
$$g^{*}(x)= \sum\limits_{j=0}^{k} x^{k-j} a_j= \sum\limits_{j=0}^{k}\theta^{j}(a_{k-j})x^{j}.$$
If $a_0 \neq 0$, the left monic skew reciprocal polynomial of $g$
is $g^{\natural}:=\frac{1}{\theta^{k}(a_0)}g^{*}$ \cite[Definition
3]{BU3}.

\begin{lem}
Let $C$ be a skew-$\theta$-$\lambda$-constacyclic code of length
$n$ over $\mathbb{F}_q.$ Then there exists a monic polynomial $g$
of minimal degree in $C$ such that $g(x)$ is a right divisor of
$x^n- \lambda$ and $C=\langle\,g(x)\,\rangle.$
\end{lem}

Let $g(x)= x^{m}+ a_{m-1}x^{m-1}+\cdots+a_0$ be a generator of a
skew-$\theta$-$\lambda$-constacyclic code of length $n$ over
$\mathbb{F}_q$ . Since $x^n- \lambda= h(x)g(x)$ for some $h \in
\mathbb{F}_q[x,\theta]$; then $a_0$ must be a non-zero element of
$\mathbb{F}_q.$ From \cite[Theorem 1]{BU2}, we have the following
result.

\begin{lem}\label{dual}
Let $C$ be a skew-$\theta$-$\lambda$-constacyclic code of length
$n$ over $\mathbb{F}_q$ generated by a monic polynomial $g$ of
degree $n-k$ with $g(x)= x^{n-k}+ \sum\limits_{j=0}^{n-k-1} a_i
x^{i}.$ Let
$\lambda^{*}=\frac{\theta^{n}(a_0)}{a_0\theta^{n-k}(\lambda)}.$
Then $C^{\perp }$ is a skew-$\theta$-$\lambda^{*}$-constacyclic
code of length $n$ over $\mathbb{F}_q$ such that $C^{\perp
}=\langle\,h^{*}(x)\,\rangle$ where $h$ is a monic polynomial of
degree $k$ such that $x^n- \theta^{-k}(\lambda)=g(x)h(x).$
Moreover $h^{*}(x)$ is a right divisor of $x^n- \lambda^{*}.$
\end{lem}

We generalize the notion of multi-twisted codes which be defined
in \cite{AH}, by introducing skew-multi-twisted codes.

\begin{definition}
Let $C$ be a linear code of length $nq$ over $\mathbb{F}_q$ and
$(\theta,\sigma)\in\texttt{Aut}(\mathbb{F}_q)\times\mathbb{S}_q.$
Let $\lambda_0, \lambda_1,\cdots,\lambda_{q-1}$ be units in
$\mathbb{F}_q.$ The code $C$ is  \emph{skew $(\lambda_0,
\lambda_1,\cdots,\lambda_{q-1})$-multi-twisted} w.r.t
$(\theta,\sigma)$ if for any codeword
$\textbf{c}:=\left(\textbf{c}_0,
\textbf{c}_1,\cdots,\textbf{c}_{q-1}\right) \in C$ where
$\textbf{c}_i:=\Phi_i(\textbf{c})$ for $0\leq i\leq q-1,$ the word
$$
\biggl(\texttt{T}_{\theta,\lambda_0}\left(\textbf{c}_{\sigma^{-1}(0)}\right),
\texttt{T}_{\theta,\lambda_1}\left(\textbf{c}_{\sigma^{-1}(1)}\right),\cdots,\texttt{T}_{\theta,\lambda_{q-1}}\left(\textbf{c}_{\sigma^{-1}(q-1)}\right)\biggr)
$$ is in $C.$
\end{definition}

\begin{theo}
Let $\lambda= \lambda_0 \eta_0+ \lambda_1 \eta_1+\cdots+
\lambda_{q-1} \eta_{q-1}$ be a unit of $R_q$ and $C$ be a linear
code of length $n$ over $R_q.$ Then $C$ is skew
$\Theta_{\theta,\sigma}$-$\lambda$-constacyclic if and only if the
linear code $\Phi(C)$ of length $qn$ over $\mathbb{F}_q$ is skew
$(\lambda_0, \lambda_1,\cdots,\lambda_{q-1})$-multi-twisted w.r.t
$(\theta,\sigma).$
\end{theo}

\begin{Proof} Set $\Theta:=\Theta_{\theta,\sigma}.$ Let $\textbf{c}= (c_0, c_1, \cdots, c_{n-1}) \in(R_q)^n$
where $\textbf{c}_{j}:=c_{0,j} \eta_0+ c_{1,j} \eta_1+ \cdots +
c_{q-1,j} \eta_{q-1};$ for all $0 \leq j \leq n-1.$ Then for any
$0\leq t \leq q-1,$

\begin{eqnarray*}
  \Phi_t(\texttt{T}_{ \Theta,\lambda}(\textbf{c})) &=& \Phi_t(\lambda\Theta(c_{n-1}), \Theta(c_0),\cdots, \Theta(c_{n-2})); \\
    &=& \left(\phi_t(\lambda\Theta(c_{n-1})), \phi_t(\Theta(c_0)),\cdots, \phi_t(\Theta(c_{n-2}))\right); \\
        &=& \left(\phi_t\left(\left(\sum\limits_{i=0}^{q-1}\lambda_i\eta_i\right)\left(\sum\limits_{i=0}^{q-1}\theta(c_{i, n-1})\eta_{\sigma(i)}\right)\right),\phi_t\left(\sum\limits_{i=0}^{q-1}\theta(c_{i, 0})\eta_{\sigma(i)}\right),\cdots, \phi_t\left(\sum\limits_{i=0}^{q-1}\theta(c_{i, n-2})\eta_{\sigma(i)}\right)\right);\\
        &=& \left(\phi_t\left(\sum\limits_{i=0}^{q-1}\lambda_i\theta(c_{\sigma^{-1}(i), n-1})\eta_{i}\right),\phi_t\left(\sum\limits_{i=0}^{q-1}\theta(c_{\sigma^{-1}(i), 0})\eta_{i}\right),\cdots, \phi_t\left(\sum\limits_{i=0}^{q-1}\theta(c_{\sigma^{-1}(i), n-2})\eta_{i}\right)\right);\\
        &=& \left(\lambda_t\theta(c_{\sigma^{-1}(t), n-1}),  \theta(c_{\sigma^{-1}(t), 0}),\cdots,  \theta(c_{\sigma^{-1}(t), n-2})\right);\\
        &=& \texttt{T}_{\theta,\lambda_t}(\textbf{c}_{\sigma^{-1}(t)}).
\end{eqnarray*}
Thus, it follows that
\begin{eqnarray*}
  \Phi(\texttt{T}_{\Theta,\lambda}(\textbf{c})) &=&\left(\Phi_0(\texttt{T}_{\Theta,\lambda}(\textbf{c})), \Phi_1(\texttt{T}_{\Theta,\lambda}(\textbf{c})),\cdots, \Phi_{q-1}(\texttt{T}_{\Theta,\lambda}(\textbf{c}))\right) ; \\
    &=&\left(\texttt{T}_{\theta,\lambda_0}(\textbf{c}_{\sigma^{-1}(0)}), \texttt{T}_{\theta,\lambda_1}(\textbf{c}_{\sigma^{-1}(1)}),\cdots, \texttt{T}_{\theta,\lambda_{q-1}}(\textbf{c}_{\sigma^{-1}(q-1)})\right);\\
    &=& \left(\texttt{T}_{\theta,\lambda_0}(\Phi_{\sigma^{-1}(0)}(\textbf{c})), \texttt{T}_{\theta,\lambda_1}(\Phi_{\sigma^{-1}(1)}(\textbf{c})),\cdots,  \texttt{T}_{\theta,\lambda_{q-1}}(\Phi_{\sigma^{-1}(q-1)}(\textbf{c}))\right).
\end{eqnarray*}
Hence, $C$ is a skew $\Theta$-$\lambda$-constacyclic code of
length $n$ over $R_q$ if and only if
$\texttt{T}_{\Theta,\lambda}(C) = C,$ which happens if and only if
$\Phi(\texttt{T}_{\Theta,\lambda}(C))= \Phi(C),$ which is
equivalent to saying $\Phi(C)$ is skew $(\lambda_0,
\lambda_1,\cdots,\lambda_{q-1})$-multi-twisted code w.r.t
$(\theta,\sigma).$
\end{Proof}

The following results are direct consequences of the above
theorem:

\begin{cor} \label{ref1}
Let $C= \eta_0 C_0 \oplus \eta_1 C_1 \oplus \cdots \oplus
\eta_{q-1}C_{q-1}$ be a linear code of length $n$ over $R_q$ and
$\lambda= \lambda_0\eta_0+\lambda_1\eta_1+
\cdots+\lambda_{q-1}\eta_{q-1}$ be a unit of $R_q.$ Then $C$ is a
skew $\Theta_{\theta,\sigma}$-$\lambda$-constacyclic code of
length $n$ over $R_q$ if and only if $C_i$ is a skew
$\theta$-$\lambda_{\sigma(i)}$-constacyclic code of length $n$
over $\mathbb{F}_q,$ for all $0 \leq i \leq q-1.$
\end{cor}

In the sequel, we only consider the automorphism
$\Theta_{\theta}(=\Theta_{\theta, \texttt{id}})$ defined by
\begin{align}\begin{array}{cccc}
\Theta_{\theta}: & R_q& \longrightarrow & R_q\\
& \sum\limits_{i=0}^{q-1}a_i\eta_i & \longmapsto&
\sum\limits_{i=0}^{q-1}\theta(a_i)\eta_i,
\end{array}
\end{align} where $\theta\in\texttt{Aut}(\mathbb{F}_q).$
Now, we give a generator of a skew-constacyclic code over $R_q.$

\begin{theo}
Let $\lambda=\lambda_0\eta_0+ \lambda_1\eta_1 +\cdots+
\lambda_{q-1}\eta_{q-1}\in\mathcal{U}(R_q)$ and $C= \eta_0 C_0
\oplus \eta_1 C_1 \oplus\cdots\oplus \eta_{q-1}C_{q-1}$ be a skew
$\Theta_{\theta}$-$\lambda$-constacyclic code over $R_q.$ Then
there exist polynomials $g_0, g_1,\cdots,g_{q-1} \in
\mathbb{F}_q[x,\theta]$ such that
$$C=\langle\,\eta_0 g_0, \eta_1 g_1,\cdots,\eta_{q-1}g_{q-1}\,\rangle$$ with
$C_i=\langle\,g_i\,\rangle$ in
$\frac{\mathbb{F}_q[x,\theta]}{\langle\,x^n- \lambda_i\,\rangle}.
$
\end{theo}

\begin{Proof} Let $D=\langle\,\eta_0 g_0, \eta_1 g_1,\cdots,\eta_{q-1}g_{q-1}\,\rangle.$ Let
$\textbf{c}(x) \in C$ such that
$\textbf{c}(x)=\sum\limits_{i=0}^{q-1} \eta_i\textbf{c}_i(x)$ with
$\textbf{c}_i \in C_i,~0 \leq i \leq q-1.$ Since
$C_i=\langle\,g_i\,\rangle$ as a left submodule of
$\frac{\mathbb{F}_q[x,\theta]}{\langle\,x^n- \lambda_i\,\rangle}
$; there exist $k_0(x), k_1(x),\cdots,k_{q-1}(x) \in
\mathbb{F}_q[x,\theta]$ such that
$\textbf{c}(x)=\sum\limits_{i=0}^{q-1} \eta_i k_i(x)g_i(x)$.
Therefore $\textbf{c} \in D.$ Reciprocally, let $\textbf{d} \in
D$, there exist $l_0(x), l_1(x),\cdots,l_{q-1}(x) \in
\frac{R_q[x,\Theta_{\theta}]}{\langle\,x^n- \lambda\,\rangle} $
such that $\textbf{d}(x)= \sum\limits_{i=0}^{q-1} \eta_i
l_i(x)g_i(x).$ Then for all $0 \leq i \leq q-1,$ there exists
$a_i(x) \in \mathbb{F}_q[x,\theta]$ such that $\eta_i l_i(x)=
\eta_i a_i(x)$, hence $\textbf{d}(x)=\sum\limits_{i=0}^{q-1}
\eta_i a_i(x)g_i(x)$  and $\textbf{d} \in C.$
\end{Proof}

\begin{cor}\label{ref3}
Under the above assumptions; let  $C= \eta_0 C_0\oplus \eta_1
C_1\oplus\cdots\oplus\eta_{q-1}C_{q-1}$ be a skew
$\Theta_{\theta}$-$\lambda$-constacyclic code over $R_q$ such that
$$C=\langle\,\eta_0 g_0, \eta_1
g_1,\cdots,\eta_{q-1}g_{q-1}\,\rangle.$$ Then $C$ is principally
generated with $C=\langle\,g(x)\,\rangle$, where $g(x)=
\eta_0g_0(x)+\eta_1g_1(x)+\cdots+\eta_{q-1}g_{q-1}(x).$ Moreover
$g(x)$ is a right divisor of $x^n- \lambda$ in $R_q[x,
\Theta_{\theta}].$
\end{cor}

\begin{Proof} It is clear that $\langle\,g(x)\,\rangle \subseteq C.$ Reciprocally, we have
$\eta_i g(x)= \eta_ig_i(x)$ for all $0 \leq i \leq q-1,$ which
implies that $C \subseteq \langle\,g(x)\,\rangle,$ whence
$C=\langle\,g(x)\,\rangle.$ Since $g_i$ is a right divisor of
$x^n- \lambda_i$, for all $0 \leq i \leq q-1,$ there exists
$h_i(x) \in \mathbb{F}_q[x, \theta]$ such that $x^n- \lambda_i=
h_i g_i.$ Since $\eta_i(x^n - \lambda)= \eta_i(x^n - \lambda_i)$
therefore
\begin{eqnarray*}
  \left(\sum\limits_{i=0}^{q-1} \eta_i h_i\right)\left(\sum\limits_{i=0}^{q-1} \eta_i g_i\right) &=& \sum\limits_{i=0}^{q-1} \eta_i h_i g_i= \sum\limits_{i=0}^{q-1} \eta_i(x^n- \lambda_i) \\
    &=& \sum\limits_{i=0}^{q-1} \eta_i(x^n- \lambda) \\
    &=& x^n- \lambda,
\end{eqnarray*}
which implies that $\sum\limits_{i=0}^{q-1} \eta_i h_i$ is a right
divisor of $x^n-\lambda$ over $R_q.$ \end{Proof}

By a similar work, and by Lemma \ref{dual} we have the
corresponding result for skew-dual codes.

\begin{lem}
Under the above assumptions; let  $C= \eta_0 C_0\oplus \eta_1
C_1\oplus\cdots\oplus\eta_{q-1}C_{q-1}$ be a skew
$\Theta_{\theta}$-$\lambda$- constacyclic code over $R_q$ such
that $C=\langle\,\eta_0 g_0, \eta_1
g_1,\cdots,\eta_{q-1}g_{q-1}\,\rangle$ with $g_i=x^{n- k_i}+
\sum_{j=0}^{k_i}a_{ij}x^{j}.$ Let $\lambda_i^{*}
=\frac{\theta^n(a_{i0})}{a_{i0}\theta^{n-k_i}(\lambda_i)}$, for
all $0 \leq i \leq q-1$ and $\lambda^{*}= \sum\limits_{i=0}^{q-1}
\eta_i \lambda_i^{*}.$ Then $C^{\perp }$ is a
skew-$\Theta_{\theta}$-$\lambda^{*}$-constacyclic codes over $R_q$
such that $C^{\perp }=\langle\,h^{*}(x)\,\rangle,$ where
$h^{*}(x)=\sum\limits_{i=0}^{q-1} \eta_i h_i^{*}$ and
$h_0,h_1,\cdots,h_{q-1}$ are skew monic polynomials of
$\mathbb{F}_q[x,\theta]$  such that
$x^{n}-\theta^{-k_i}(\lambda_i)=g_i h_i$ in $\mathbb{F}_q[x,
\theta].$ Moreover $h^{*}(x)$ is a right divisor of $x^n-
\lambda^{*}$ in $R_q[x, \Theta_{\theta}].$
\end{lem}

Now we characterize self-dual skew
$\Theta_{\theta}$-$\lambda$-constacyclic codes of length $n$ over
$R_q.$

\begin{theo}\label{ref2}
Let $C$ be a linear code of even length $n$ over $R_q$ and
$\lambda$ be a unit of $R_q.$ If $C$ is a self-dual skew
$\Theta_{\theta}$-$\lambda$-constacyclic code of length $n$ over
$R_q$ then either $C$ is skew $\Theta_{\theta}$-cyclic or  skew
$\Theta_{\theta}$-negacyclic.
\end{theo}

\begin{Proof} Let $C= \eta_0 C_0 \oplus \eta_1 C_1 \oplus \cdots \oplus \eta_{q-1}C_{q-1}$ be
a skew $\Theta_{\theta}$-$\lambda$-constacyclic code of length $n$
over $R_q.$ It is follows from Proposition \ref{pro:Cself} that
$C$ is self-dual code over $R_q$ if and only if $C_0,
C_1,\cdots,C_{q-1}$ are self-dual codes over $\mathbb{F}_q.$ From
Corollary \ref{ref1}, $C$ is a skew
$\Theta_{\theta}$-$\lambda$-constacyclic code of length $n$ over
$R_q$ if and only if $C_i$ is a skew
$\theta$-$\lambda_i$-constacyclic code of length $n$ over
$\mathbb{F}_q$; for all $0 \leq i \leq q-1.$ But the only
self-dual skew $\theta$-constacyclic codes over $\mathbb{F}_q$ are
self-dual skew $\theta$-cyclic codes and self-dual skew
$\theta$-negacyclic codes \cite[Proposition 5]{BU3}. Moreover,
from  \cite[Proposition 3]{B}, there cannot exist both a self-dual
skew $\theta$-cyclic code and a self-dual skew $\theta$-negacyclic
code with the same dimension over $\mathbb{F}_q.$ Then all $C_i$
are skew $\theta$-cyclic codes or skew $\theta$-negacyclic codes
of length $n.$ Whence $C$ is a self-dual skew
$\Theta_{\theta}$-cyclic code or a self-dual skew
$\Theta_{\theta}$-negacyclic code of length $n$ over $R_q.$
\end{Proof}

The following result gives necessary and sufficient conditions for
the existence of self-dual skew cyclic and self-dual skew
negacyclic codes over a finite field \cite[Propositions 4 and
5]{B}.

\begin{pro}
Let $\mathbb{F}_{q}$ be a finite field with characteristic $p,$
which is an odd prime and $q=p^r.$ Let $\theta$ be an
$\mathbb{F}_{p^s}$-automorphism of $\mathbb{F}_q,$ where $s$
divides $r.$
\begin{enumerate}
\item If $q \equiv 1 \mod 4$, then:
\begin{itemize}
\item[$i)$] there exists a self-dual skew cyclic code of dimension
$k$ over $\mathbb{F}_{q}$ if and only if $p \equiv 3 \mod 4$, $r$
is even and $s \cdot k$ is  odd.
\item [$ii)$]there exists a
self-dual skew negacyclic code of dimension $k$ over
$\mathbb{F}_{q}$ if and only if $p \equiv 1 \mod 4$ or $p \equiv 3
\mod 4$, $r$ is even and $s \cdot k$ is even.
\end{itemize}
\item  If $q \equiv 3 \mod 4$, then:
\begin{itemize}
\item [$i)$] There does not exist a self-dual skew cyclic code of
dimension $k$ over $\mathbb{F}_{q}.$
\item [$ii)$]There exists a
self-dual skew negacyclic code of dimension k over
$\mathbb{F}_{q}$ if and only if $k \equiv 0 \mod 2^{\mu-1}$, where
$\mu \geq 2$ is the biggest integer such that $2^{\mu}$ divides
$p+1.$
\end{itemize}
\end{enumerate}

\end{pro}

From Theorem \ref{ref2}, the previous result is extend naturally
to the ring $R_q.$

\begin{pro}
Let $\mathbb{F}_{q}$ be a finite field with characteristic $p,$
which is an odd prime and $q=p^r.$ Let $\theta$ be an
$\mathbb{F}_{p^s}$-automorphism of $\mathbb{F}_q,$ where $s$
divides $r.$
\begin{enumerate}
\item If $q \equiv 1 \mod 4$, then:
\begin{itemize}
\item[$i)$] there exists a self-dual skew $\Theta_{\theta}$-cyclic
code of dimension $k$ over $R_{q}$ if and only if $p \equiv 3 \mod
4$, $r$ is even and $s \cdot k$ is  odd.
\item [$ii)$]there exists
a self-dual skew $\Theta_{\theta}$-negacyclic code of dimension
$k$ over $R_{q}$ if and only if $p \equiv 1 \mod 4$ or $p \equiv 3
\mod 4$, $r$ is even and $s\cdot k$ is even.
\end{itemize}
\item  If $q \equiv 3 \mod 4$, then:
\begin{itemize}
\item [$i)$] There does not exist a self-dual skew
$\Theta_{\theta}$-cyclic code of dimension $k$ over $R_{q}.$
\item [$ii)$]There exists a self-dual skew
$\Theta_{\theta}$-negacyclic code of dimension k over $R_{q}$ if
and only if $k \equiv 0 \mod 2^{\mu-1}$, where $\mu \geq 2$ is the
biggest integer such that $2^{\mu}$ divides $p+1.$
\end{itemize}
\end{enumerate}
\end{pro}

\begin{exa}
Let $\mathbb{F}_4=\{0,1,\alpha, \alpha^2\}$ be a finite field with
$\alpha^2+\alpha+1=0.$ Let $\theta$ be the Frobenius automorphism
$(a \longmapsto a^2)$ over $\mathbb{F}_4.$ We use the technique
given in \cite{BU3} to construct self-dual skew-$\theta$-cyclic
codes of length 6 over $\mathbb{F}_4.$ The factorization of
$y^3-1$ into irreducible monic polynomials over $\mathbb{F}_2$ is
given by: $y^3-1=(y+1)(y^2+y+1),$ and we set $f_1(y):=y+1$ and
$f_2(y):=y^2+y+1.$ We have also $f_1^{\natural}(y)=f_1(y)$ and
$f_2^{\natural}(y)=f_2(y).$

Let $\mathcal{G}_1:=\{g \in \mathbb{F}_4[x, \theta]:
g^{\natural}g= x^2+1\}$  and $\mathcal{G}_2:=\{g \in
\mathbb{F}_4[x, \theta]: g^{\natural}g= x^4+x^2+1\}.$ We compute
$\mathcal{G}_1$ and $\mathcal{G}_2$ by using Gröbner basis and
Magma language system \cite{BCP}:
$$\mathcal{G}_1:=\{x+1\};~\mathcal{G}_2:=\{x^2+x+1, x^2+ \alpha, x^2+\alpha^2\}.$$
The polynomial $g_1=\texttt{lclm}(x+1, x^2+x+1)=x^3+1$ generates a
$[6,3,2]$-self-dual skew $\theta$-cyclic code  and
$g_2=\texttt{lclm}(x+1, x^2+\alpha^2)=x^3+ \alpha x^2+ \alpha
x+1$; $g_3=\texttt{lclm}(x+1, x^2+\alpha)=x^3+ \alpha^2 x^2+
\alpha^2 x+1$ generate $[6,3,3]$-self-dual skew $\theta$-cyclic
codes over $\mathbb{F}_4.$

 Let $\eta_0=\overline{v}^3+1; \eta_1= \overline{v}^3+\overline{v}+1;
\eta_2=\overline{v}^3+\alpha \overline{v}^2+\alpha^2\overline{v};
\eta_4=\overline{v}^3+\alpha^2\overline{v}^2+\alpha \overline{v}$
be the complete set of primitive pairwise orthogonal idempotents
of $R_4.$ Let $C= \eta_0 C_0 \oplus \eta_1 C_1 \oplus\eta_2
C_2\oplus\eta_{3}C_3$ with $C_0, C_1,C_2,C_3$ be  the skew
$\theta$-cyclic codes over $\mathbb{F}_4$ generated by $g_1, g_1,
g_2, g_3$ respectively. According to Corollary \ref{ref3}, $C$ is
a $[6,3,2]$-self-dual skew $\Theta_{\theta}$-cyclic code over
$R_4$ generated by $g=\eta_0g_1+\eta_1g_1+\eta_2g_2+\eta_3g_3=
x^3+(\overline{v}^3+\overline{v}^2)x^2+(\overline{v}^3+\overline{v}^2)x+1.$
A generator matrix of $C$ is given by:
$$G:=\left(\begin{array}{cccccc}
1&\overline{v}^3+\overline{v}^2&\overline{v}^3+\overline{v}^2&1&0&0\\
0&1&\overline{v}^3+\overline{v}^2&\overline{v}^3+\overline{v}^2&1&0\\
0&0&1&\overline{v}^3+\overline{v}^2&\overline{v}^3+\overline{v}^2&0
\end{array}
\right). $$
\end{exa}

\section{Conclusion}

In this work, skew constacyclic codes are investigated over the
non-chain ring $R_q.$ The automorphism group of this ring is
given. We defined a linear Gray map and established a connection
between skew constacyclic codes of length $n$ over $R_q$ and skew
multi-twisted codes of length $qn$ over $\mathbb{F}_q.$ We
determined generator polynomials for skew constacyclic codes over
$R_q$ and for their dual codes. We also characterized self-dual
skew constacyclic codes and provided necessary and sufficient
conditions for the existence of self-dual skew cyclic and
self-dual skew negacyclic codes over $R_q.$

\end{document}